\documentclass[conference]{IEEEtran}
\usepackage{amsmath, amsthm}
\usepackage{color}

\usepackage{times}
\usepackage[latin1]{inputenc}
\usepackage{amsfonts}
\usepackage{psfrag}
\usepackage{amsbsy}
\usepackage{amssymb}
\usepackage[mathscr]{eucal}
\usepackage{latexsym}
\usepackage[T1]{fontenc}
\usepackage[dvips]{graphicx}
\usepackage{bm}
\usepackage{array}
\usepackage{tikz}
\usetikzlibrary{arrows,automata}
\usetikzlibrary{dsp,chains}
\usetikzlibrary{arrows,automata}
\usetikzlibrary{shapes,patterns}
\usetikzlibrary{dsp,chains}
\usepackage{acronym}
\usepackage{balance}

\DeclareMathAlphabet{\mathpzc}{OT1}{pzc}{m}{it}

\usepackage{verbatim,mathrsfs, cite}

\makeatletter
\newif\ifAC@uppercase@first
\def\Aclp#1{\AC@uppercase@firsttrue\aclp{#1}\AC@uppercase@firstfalse}
\def\AC@aclp#1{%
  \ifcsname fn@#1@PL\endcsname
    \ifAC@uppercase@first
      \expandafter\expandafter\expandafter\MakeUppercase\csname fn@#1@PL\endcsname
    \else
      \csname fn@#1@PL\endcsname
    \fi
  \else
    \AC@acl{#1}s
  \fi 
}
\edef\AC@uppercase@write{\string\ifAC@uppercase@first\string\expandafter\string\MakeUppercase\string\fi\space}
\def\AC@acrodef#1[#2]#3{%
  \@bsphack
  \protected@write\@auxout{}{%
    \string\newacro{#1}[#2]{\AC@uppercase@write #3}%
  }\@esphack
}
\def\Acl#1{\AC@uppercase@firsttrue\acl{#1}\AC@uppercase@firstfalse}
\makeatother

\acrodef{HARQ}{hybrid automatic repeat request}
\acrodef{ARQ}{automatic repeat request}
\acrodef{CSI}{channel state information}
\acrodef{SNR}{signal to noise ratio}
\acrodef{ACK}{acknowledge}
\acrodef{NACK}{not acknowledge}
\acrodef{CRC}{code redundancy check}
\acrodef{MIMO}{multiple input multiple output}
\acrodef{S-HARQ}{secure \ac{HARQ}}
\acrodef{CDF}{cumulative distribution function}
\acrodef{IR-HARQ}{incremental redundancy \ac{HARQ}}

\newtheorem{lemma}{Lemma}
\newtheorem{theorem}{Theorem}

\newcommand\nframe{K}
\newcommand\iframe{k}
\newcommand\nslot{M}
\newcommand\islot{m}

\begin{document}
\sloppy
\title{Secret Message Transmission by HARQ \\ with Multiple Encoding}

\author{Stefano Tomasin and Nicola Laurenti \\
\small Department of Information Engineering,
University of Padova, Via Gradenigo 6/B, I-35131 Padova, Italy\\
\{tomasin, nil\}@dei.unipd.it
}

\maketitle

\begin{abstract}
Secure transmission between two agents, Alice and Bob, over block fading channels can be achieved similarly to conventional \ac{HARQ} by letting Alice transmit multiple blocks, each containing an encoded version of the secret message, until Bob informs Alice about successful decoding by a public error-free return channel. In existing literature each block is a differently punctured version of a single codeword generated with a Wyner code that uses a common randomness for all blocks. In this paper instead we propose a more general approach where multiple codewords are generated from independent randomnesses. The class of channels for which decodability and secrecy is ensured is characterized, with derivations for the existence of secret codes. We show in particular that the classes are not a trivial subset (or superset) of those of existing schemes, thus highlighting the novelty of the proposed solution. The result is further confirmed by deriving the average achievable secrecy throughput, thus taking into account both decoding and secrecy outage.
\end{abstract}
\acresetall
\acresetall

\section{Introduction}
Physical layer secrecy has gained a lot of attention in the last few years, due to its ability of providing information theoretic unconditional security, thus adding security at the physical layer. From the seminal works \cite{Wyner1975, Csiszar1978, Leung-Yan-Cheong1978} performance limits and achievable rates have been derived in different scenarios for the reliable yet secret transmission of confidential information (see, e.g. \cite{Bloch} for a review). In particular, it has been shown that diversity, in the form of fading (temporal diversity), multipath (frequency diversity) or \ac{MIMO} (spatial diversity) is definitely beneficial to secret transmission. In fact, dimensions or instants in which the legitimate receiver is at an advantage with respect to the eavesdropper can be selected, even if the channels to both receivers have the same statistics. However, one of the main obstacles to the effective implementation of such systems is the need of knowing \ac{CSI} towards both the legitimate receiver and the eavesdropper at the time of code design. 
 
This drawback can be partly mitigated by the presence of a feedback channel, even with a limited rate and/or publicly accessible, which, contrary to the unconstrained transmission case, has been shown to increase secrecy capacity \cite{Ardestanizadeh2009}.
The \ac{ARQ} mechanism, with its intrinsic one-bit feedback is leveraged in \cite{Abdallah2011} for the secure generation of cryptographic keys, that can either be used in traditional encryption systems or provide perfect secrecy through one-time-pad schemes.
Instead of considering a simple retransmission approach, a \ac{HARQ}-like scheme is derived in \cite{Tang2009} for secure communications over a block-fading Gaussian channel. In this case, a single codeword is generated and punctured versions of it are transmitted until the legitimate receiver decodes the secret message. For encoding, a wiretap code with incremental redundancy is employed, obtaining an \ac{IR-HARQ} scheme, and an outage formulation is considered. The approach of \cite{Tang2009} is then extended in \cite{Sarikaya2012} to a multiuser uplink scenario, with each user aiming at transmitting a combination of public and confidential messages to a single base station, with the other users acting potentially as eavesdroppers. A suboptimal strategy of power allocation and scheduling to maximize the overall network utility is then derived. The \ac{HARQ} secrecy scenario is also considered in \cite{Baldi2012} with a rather different approach, where standard codes are used with the addition of scrambling, but secrecy is expressed in terms of the bit error probability at the eavesdropper. Lastly, the throughput of \ac{HARQ} without secrecy constraints has been studied in \cite{Caire-jul01}, but the analysis does not fit immediately the different scenario where a secret message must be transmitted. 
 
In this paper we consider a scenario similar to \ac{IR-HARQ} of \cite{Tang2009}, where transmissions occur on block fading channels. While for \ac{IR-HARQ} at each retransmission a different puncturing of a single codeword of a Wyner code is used, we propose to encode the secret message with multiple codes and then send punctured versions of the multiple codewords until Bob decodes. The obtained solution is denoted \ac{S-HARQ}, and is a strict generalization of \ac{IR-HARQ}. We prove the existence of codes that ensure both decodability and secrecy for all channel realizations that satisfy certain conditions. While decodability and secrecy in \ac{IR-HARQ} depend on the {\em average over all channel realization} of the mutual information between the transmitted and received messages, in \ac{S-HARQ} instead performance depends on {\em multiple averages} over different fading blocks, thus providing additional degrees of freedom\footnote{Notation, throughout the paper ${\rm E}_X[g(X)]$ denotes the expectation of $g(X)$ with respect to $X$, and ${\rm I}(X; Y)$ denotes the mutual information between $X$ and $Y$.}.

\section{System Model}

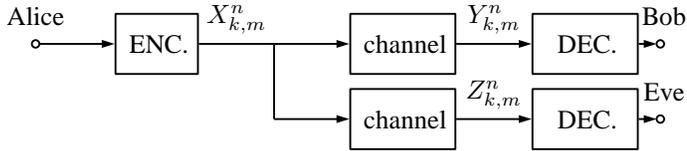
\begin{figure}
\centering
\begin{tikzpicture}
	\node[dspnodeopen] (Al) {Alice};
	\node[dspsquare, right= of Al]       (c0) {\, ENC. };
	\node[coordinate,right= of c0]       (c1) {};	
	\node[dspsquare,right= of c1]                    (c2) {\, channel };
	\node[dspfilter,right=of c2]                     (c3) {\, DEC. };
	\node[dspnodeopen, right of=c3] (A) {\;Bob};
	\node[coordinate,below= of c1]       (c1b) {};	
	\node[dspsquare,right= of c1b]                    (c4) {\, channel };
	\node[dspfilter,right=of c4]                     (c5) {\, DEC. };
	\node[dspnodeopen, right of=c5] (E) {\;Eve};
    \draw[dspconn] (Al) --  (c0);
    \draw[dspline] (c0) -- node[midway, above] {$X_{\iframe,\islot}^n$} (c1) ;
    \draw[dspconn] (c1) -- (c2);
    \draw[dspconn] (c2) -- node[midway,above] {$Y_{\iframe,\islot}^n$}(c3);
    \draw[dspline] (c1) -- (c1b);
    \draw[dspconn] (c1b) -- (c4);
    \draw[dspconn] (c4) -- node[midway,above] {$Z_{\iframe,\islot}^n$} (c5);
    \draw[dspconn] (c3) -- (A);
    \draw[dspconn] (c5) -- (E);
    
\end{tikzpicture}
\caption{System Model.}
\label{fig1}
\end{figure}

We consider the scenario of Fig. \ref{fig1}, where an agent Alice transmits a secret message $\mathcal M$ to an intended destination agent, Bob, over the Alice-Bob channel.  A third agent, Eve, overhears the message transmitted by Alice over an independent Alice-Eve channel. 

Time is organized in consecutive {\em slots} of the same duration, which are grouped into {\em frames}, each comprising $\nslot$ consecutive slots. The transmission of $\mathcal M$ spans in general many slots. At slot $\islot$ of frame $\iframe$, Alice transmits message $X_{\iframe,\islot}^n$ of $n$ symbols containing an encoded version of the secret message. The secret message $\mathcal M$ bears $n R_s$ bits of information. At each slot, after Alice's transmission, Bob sends a \ac{NACK} message if he fails to correctly decode the secret message. When decoding is successful, Bob sends an \ac{ACK} feedback to Alice, who stops transmissions. After $\nslot\nframe$ slots used for the transmission of the same secret message, Alice discards the secret message, irrespective of Bob decoding outcome. ACK/NACK messages are perfectly received by Eve. The procedure is similar to \ac{HARQ}, except for the presence of Eve, thus we denote the resulting scheme as \ac{S-HARQ}. 

The message received by Bob at slot $\islot$ of frame $\iframe$ is denoted by $Y_{\iframe,\islot}^n$, which is related to $X_{\iframe,\islot}^n$ by a given transition probability distribution. Let us also define $X_\iframe^{nM} = [X_{\iframe,1}^n, \ldots, X_{\iframe,\nslot}^n]$ and $Y_\iframe^{nM} = [Y_{\iframe,1}^n, \ldots, Y_{\iframe,\nslot}^n]$. Eve overhears the packet $Z_{\iframe,\islot}^n$ at slot $\islot$ of frame $\iframe$, and let us define $Z_\iframe^{nM} = [Z_{\iframe,1}^n, \ldots, Z_{\iframe,\nslot}^n]$. Both Bob and Eve use all previously received packets for decoding purposes. Moreover, Eve knows the encoding procedure followed by Alice.

We assume a block-fading channel, i.e., the channels do not change within each slot, while they vary from slot to slot. We also assume that the number of channel uses $n$ within each slot is large enough so that we can use random coding arguments. For the Alice-Bob channel  $\{p(Y_{\iframe,\islot}|X_{\iframe,\islot})\}$ denotes the symbol transition probability for slot $\islot=1,2, \ldots, \nslot$, and frame $\iframe=1,2, \ldots, \nframe$. Similarly for the Alice-Eve channel we have $\{p(Z_{\iframe,\islot}|X_{\iframe,\islot})\}$. Let $p_{\islot, \iframe} = \{ p(Y_{\iframe,\islot}|X_{\iframe,\islot}), p(Z_{\iframe,\islot}|X_{\iframe,\islot})\}$ denote a generic channel realization for slot $\islot$ and frame $\iframe$, while $p_k = \{ p_{\islot, k},\; \islot=1,2, \ldots, \nslot\}$ is a set of channel realizations in frame $k$ and 
\begin{equation}
\bm{p} = \{ p_{\islot, \iframe},\; \islot=1,2, \ldots, \nslot, \iframe=1,2, \ldots, \nframe\}
\end{equation}
denotes a generic channel realization for the whole sequence of frames. Fading implies that $\bm{p}$ is a random vector, and block fading statistics dictates the statistics of the vector. We also assume that \ac{CSI} is not available to Alice before transmission.

\section{Code Construction}

Encoding aims at ensuring both decodability of $\mathcal M$ by Bob, and secrecy, i.e., preventing information leakage to Eve. To this end Alice encodes the secret message by using a random binning approach, independently (in the random message) for each frame. The {\em random message}  used at frame $\iframe$ to confuse Eve is denoted as $\mathcal M_{{\rm d}, \iframe}$ and bears $n R_{{\rm d}, \iframe}$ bits of information. The random message is independently generated at each frame, contains no secret information and may even be completely irrelevant for the three agents, since its purpose is only to confuse Eve about $\mathcal M$. We denote the collection of the random messages over the $\nframe$ frames as ${\mathcal M}_e = \{\mathcal M_{{\rm d}, 1},  \ldots, \mathcal M_{{\rm d}, \nframe}\}$.

\paragraph*{Encoding process} Let $i\in$ be the index of the (random) secret message $\mathcal M$ in the set of $2^{nR_s}$ messages. The encoding process works as follows: at frame $\iframe$, Alice selects an index $j_\iframe \in \{0, 1, \ldots, nR_{{\rm d},\iframe}\}$ randomly and uniformly, and generates the codeword $X^{nM}_{\iframe}(i,j_\iframe)$, which is punctured into $M$ codewords of $n$ symbols $X^n_{\iframe,1}(i,j_\iframe),\ldots,X^n_{\iframe,\nslot}(i,j_\iframe)$. Then, the punctured codeword $X_{\iframe,\islot}^n(i,j_\iframe)$ is transmitted at the $\islot$-th slot of frame $\iframe$. The set of transmitted messages over all $\nframe$ frames is identified by the $(\nframe+1)-$ple $(i,j_1, \ldots, j_\nframe)$. 

\paragraph*{Decoding strategy} We consider a joint typicality decoder for both Bob and Eve. For the generic slot $\islot$ of frame $\iframe$, let $\mathcal T_{\iframe,\islot, \epsilon}^{n}(XY)$ denote the set of all  $\epsilon$-jointly weakly typical sequences $(\{X_{1,1}^n(i,j_1), \ldots, X_{\iframe,\islot}^n(i,j_\iframe)\}, \{Y_{1,1}^n, \ldots, Y_{\iframe,\islot}^n\})$. Bob decides for message $(\hat{\imath}, \hat{\jmath}_1, \ldots, \hat{\jmath}_\iframe)$ if $\{X_{\cdot, \cdot}^n(\hat\imath,\hat\jmath_\iframe)\}$ is the only sequence taken from $C_n$ that is $\epsilon$-jointly typical with $\{Y_{\cdot, \cdot}^n\}$. Otherwise an error is output. %Similarly, Eve decides for $(\hat{i}, \hat{j}_1, \ldots, \hat{j}_\iframe)$ if $(\{X_{\cdot, \cdot}^n(i,j_\iframe)\}, \{Z_{\cdot, \cdot}^n\}) \in \mathcal T_{\iframe,\islot,\epsilon}^{n}(XZ)$. 

The code used at each frame is called frame code, while  {\em \ac{HARQ} code} denotes the sequence of the frame codes. In particular:
\paragraph*{frame code} this is a subset $C_{n,\iframe}$ of $2^{n(R_s + R_{{\rm d}, \iframe})}$ words of $n\nslot$ symbols, randomly chosen. In fact, for each index pair $(a,b)$, $a\in\{1,\ldots,2^{nR_s}\}, b\in\{1,\ldots,2^{nR_{{\rm d},k}}\}$ we choose the word $X_{\iframe}^{nM}(a,b)$ with independent symbols all drawn from a same distribution $p_{X_{\iframe}}(\cdot)$;
\paragraph*{\ac{HARQ} code} this is the set of $2^{n(R_s+\sum_{\iframe=1}^\nframe R_{{\rm d}, \iframe})}$ codewords obtained by concatenating the $\nframe$ words $X^{n\nframe\nslot}(a,b_1,\ldots,b_\nframe) = [X_{1}^{n\nslot}(a,b_1),\ldots,X_{\nframe}^{n\nslot}(a,b_\nframe)]$ and is denoted by $C_n$. 

We also denote the set of all possible codes that can be generated for frame $\iframe$ as the ensemble $\mathcal C_{n,\iframe} = \{C_{n,\iframe}\}$, and that of all possible \ac{HARQ} codes as $\mathcal C_n = \{C_{n}\}$. We assume that the actually selected code (as well as the ensemble) is known to both Bob and Eve.

Note that in the case of a single frame ($\nframe =1$) we obtain the scenario considered in \cite{Tang2009}, where a single codeword is split into $\nslot$ parts that are sequentially transmitted until Bob decodes. On the other hand, when $\nslot = 1$ we have that a new codeword, generated by an independent random message, is transmitted at each slot. Other cases ($\nframe > 1$ and $\nslot > 1$) correspond to intermediate situations where each random message spans multiple slots, and more than one random message may be used to confuse the eavesdropper about the same secret message, provided that Bob needs more retransmissions. In the following we detail the general case for any value of $\nframe$ and $\nslot$.

\section{Decodability and Secrecy Conditions}

The design of the \ac{S-HARQ} code aims at ensuring that  a) Bob is able to decode the secret message with vanishing probability ({\em decodability}), and b) Eve gets vanishing information rate on the secret message ({\em secrecy}). Now we show that asymptotically ($n \rightarrow \infty$) for a given set $\mathcal P$ of channels, and for a given set of rates $(R_s, R_{{\rm d}, 1}, \ldots, R_{{\rm d}, K})$, there exists a \ac{S-HARQ} code that provides both decodability and secrecy. Note that the code to be used is the same for all channels in the set $\mathcal P$. In the considered scenario, no \ac{CSI} is available to Alice, therefore if the channel is not in $\mathcal P$ we may have an outage event, i.e., either Bob may not decode $\mathcal M$ or Eve may get some information on $\mathcal M$. The outage probability $P_{\rm out}$ is the probability that an outage event occurs. From the definition of set $\mathcal P$ we have a bound on $P_{\rm out}$ as 
\begin{equation}
P_{\rm out} \leq {\rm P}[ \bm{p} \notin \mathcal P]\,,
\end{equation}
where ${\rm P}[\cdot]$ denotes the probability operator. The characterization of the set $\mathcal P$ can then guide the code design and its usage, since we can obtain an estimate of the outage probability by assessing the probability that the channel over which the code is actually used is outside $\mathcal P$. 

In order to characterize  $\mathcal P$ we first derive conditions on the realization $\bm{p}$ that ensure decodability by Bob on average over a set of codes, then we derive conditions on $\bm{p}$ that ensure secrecy with respect to Eve on average over a set of codes. Finally we characterize the set $\mathcal P$ over which a single code provides both secrecy and decodability. 

From now on, for the sake of compactness we denote by $I_{\iframe,\islot}^{\rm B}(p_{\iframe,\islot}) = \lim_{n\to\infty}\frac1n {\rm I}(X_{\iframe,\islot}^n;Y_{\iframe,\islot}^n|p_{\iframe,\islot})$ the single letter mutual information across the legitimate channel at slot $\islot$ of frame $\iframe$, and the analogous for the eavesdropper channel by $I_{\iframe,\islot}^{\rm E}(p_{\iframe,\islot}) = \lim_{n\to\infty}\frac1n I(X_{\iframe,\islot}^n;Z_{\iframe,\islot}^n|p_{\iframe,\islot})$.

\subsection{Decodable codes ensemble characterization}

As decodability is concerned, we have the following result:
\begin{lemma}
\label{l5}
Let $(\hat{\mathcal M}, \hat{\mathcal M}_{\rm e})$ be the message decoded by the $\epsilon$-joint typicality decoder over $\nframe$ frames, and let the error probability associated with a given HARQ code $C_n$ be ${\rm P}_e(C_n|\bm{p}) = {\rm P} [({\mathcal M}, {\mathcal M}_e) \neq (\hat{\mathcal M}, \hat{\mathcal M}_e)| C_n, \bm{p}]$ for a given channel realization $\bm{p}$. For all $\nframe'\leq \nframe$ and $\nslot' \leq \nslot$ that satisfy
\begin{equation}
\begin{split}
\sum_{\iframe=1}^{\nframe'}  R_{{\rm d}, \iframe} + R_s < \sum_{\iframe=1}^{\nframe'-1} \sum_{\islot=1}^\nslot \left[{\rm I}_{\iframe,\islot}^{\rm B}(p_{\iframe,\islot}) - \delta(\epsilon)\right] + \\ 
\sum_{\iframe=1}^{\nframe'} \sum_{\islot=1}^{\nslot'} \left[{\rm I}_{\iframe,\islot}^{\rm B}(p_{\iframe,\islot}) - \delta(\epsilon)\right]\,, 
\end{split}
\label{conR1}
\end{equation}
with $\delta(\epsilon) > 0$, then for each and $n$ there exists a $\delta'_\epsilon(n)$ such that $\delta'_\epsilon(n) \xrightarrow[n \rightarrow \infty]{} 0$ for each $\epsilon$, and 
\begin{equation}
{\rm E}_{C_n}[{\rm P}_e(C_n|\bm{p})] \leq \delta_\epsilon(n)\,.
\end{equation}
\end{lemma}
\begin{proof}
See the Appendix.
\end{proof}

%\begin{lemma}
%Under the same assumptions as Lemma \ref{l5}, there exists a sequence of HARQ codes $C^\star_n$, $n\in \mathbb N$, such that for all channels in which (\ref{conR1}) holds, we have
%\begin{equation}
%{\rm P}_e(C_n) \leq \delta_\epsilon(n)\,.
%\end{equation}
%
%\end{lemma}
%Since the error probability is non-negative, when its expectation goes to zero means that there exists a sequence of codes having vanishing error probability as $n \rightarrow \infty$.

\subsection{Secrecy codes ensemble characterization}

To deal with secrecy, we first denote the {\em information leakage of the first $\nslot'$ slots of the $\iframe$-th frame} to Eve when Alice uses code $C_{n,\iframe}$ over channel realization $p_{\iframe}$ as 
\begin{equation}
{\rm L}(C_{n,\iframe}|p_{\iframe},\nslot') = {\rm I}({\mathcal M}; \{Z_{\iframe, 1}^n, \ldots, Z_{\iframe, \nslot'}^n\}| C_{n,\iframe}, p_{\iframe}).
\end{equation}
Similarly, the information leakage for the transmission up to slot $\nslot'$ of frame $\nframe'$ is defined as 
\begin{equation}
\begin{split}
{\rm L}(C_n|\bm{p},\nslot',\nframe') = \\
{\rm I}({\mathcal M}; Z^{n\nslot}_1,\ldots,Z^{n\nslot}_{\nframe'-1}, \{Z^n_{\nframe', 1}, \ldots, Z^n_{\nframe', \nslot'}\}| C_n, \bm{p}).
\end{split}
\end{equation}

Then we start with the following lemma that establishes a relation between the information leakage of each frame and that of the transmission up to frame $\nframe'$.
\begin{lemma}
\label{l3}
The information leakage over all frames up to slot $\nslot'$ of frame $\nframe'$ is not larger than the sum of information leakage for each frame, i.e.,
\begin{equation*}
{\rm L}(C_n|\bm{p},\nslot',\nframe') \leq \sum_{\iframe=1}^{\nframe'-1} {\rm L}(C_{n,\iframe}| p_{\iframe},\nslot) + {\rm L}(C_{n,\nframe'}| p_{\nframe'},\nslot')\,.
\end{equation*}
\end{lemma}
\begin{proof} For the sake of a simpler notation we provide the proof for $K' =K$ and $M'=M$, the generalization being straightforward. 

Since we use independent random binning in each transmission,  $(Z_1^{n\nslot}, \ldots, Z_{\iframe-1}^{n\nslot}, C_n) \rightarrow (C_{n,\iframe}, {\mathcal M}) \rightarrow {Z}^{n\nslot}_\iframe$ is a Markov chain. 

By the chain rule for mutual information \cite[eq. (2.62)]{Cover} we have 
\begin{equation}
\begin{split}
{\rm I}({\mathcal M}; Z^{n\nslot}_1,\ldots,{Z}_\nframe^{n\nslot} | C_n, \bm{p}) =  \\
\sum_{\iframe=1}^\nframe {\rm I}({\mathcal M}; {Z}_\iframe^{n\nslot} |Z^{n\nslot}_1,\ldots,{Z}_{\iframe-1}^{n\nslot},  C_n, p_{\iframe})
\end{split}
\end{equation}
Each term in the sum can be upper bounded as 
\begin{equation}
\begin{split}
{\rm I}({\mathcal M}; & {Z}_\iframe^{n\nslot} |Z^{n\nslot}_1,\ldots,{Z}_{\iframe-1}^{n\nslot},  C_n, p_{\iframe}) = \\
 = {} &  {\rm H}(Z_\iframe^{n\nslot} |Z^{n\nslot}_1,\ldots,{Z}_{\iframe-1}^{n\nslot},  C_n, p_{\iframe}) -  \\
 & - {\rm H}(Z_\iframe^{n\nslot} |{\mathcal M}, Z^{n\nslot}_1,\ldots,{Z}_{\iframe-1}^{n\nslot} , C_n, p_{\iframe}) \\
 \leq {} &  {\rm H}(Z_\iframe^{n\nslot} |{\mathcal M}, C_{n,\iframe}, p_{\iframe}) - \\
& - {\rm H}(Z_\iframe^{n\nslot} |{\mathcal M}, Z^{n\nslot}_1,\ldots,{Z}_{\iframe-1}^{n\nslot} , C_n, p_{\iframe}) \\
 = {} &   {\rm H}(Z_\iframe^{n\nslot} | C_{n,\iframe}, p_{\iframe}) - {\rm H}(Z_\iframe^{n\nslot} |{\mathcal M}, C_{n,\iframe}, p_{\iframe}) \\
 = {} & {\rm I}({\mathcal M}; Z_\iframe^{n\nslot} | C_{n,\iframe}, p_{\iframe}) = {\rm L}(C_{n,\iframe}| p_{\iframe}, \nslot)\,.
\end{split}
\end{equation}
\end{proof}

Then we derive a bound on the information leakage at the $\iframe$-th frame by the following lemma. 
\begin{lemma}
\label{l1}
For each channel realization $\bm{p}$ and ($\nframe'$, $\nslot'$) such that $\sum_{\islot=1}^{\nslot}{\rm I}(X_{\iframe,\islot}^n;Z_{\iframe,\islot}^n | p_{\iframe,\islot}) < R_{{\rm d}, \iframe}$, $k=1, \ldots, \nframe'-1$, $\sum_{\islot=1}^{\nslot'}{\rm I}(X_{\nframe',\islot}^n;Z_{\nframe',\islot}^n | p_{\nframe',\islot}) < R_{{\rm d}, \nframe'}$, and for each $n$ and $\epsilon$ we have a $\delta(\epsilon)$ and a $\delta_\epsilon(n)$ such that $\delta_\epsilon(n) \xrightarrow[n \rightarrow \infty]{} 0$ and 
\begin{equation*}
{\rm E}_{C_{n,\iframe}} \left[ \frac{1}{n} {\rm L}(C_{n,\iframe}| p_{\iframe}, \nslot) \right] \leq \delta(\epsilon) + \delta_\epsilon(n)\,, \quad \iframe=1,\ldots, \nframe'-1
\end{equation*}
\begin{equation*}
{\rm E}_{C_{n,\nframe'}} \left[ \frac{1}{n} {\rm L}(C_{n,\nframe'}| p_{\nframe'}, \nslot') \right] \leq \delta(\epsilon) + \delta_\epsilon(n)\,.
\end{equation*}
\end{lemma}
\begin{proof}
See the wiretap coding theorem \cite[pg. 72]{Bloch}.
\end{proof}

By combining the two results we obtain the following lemma.

\begin{lemma}
\label{l4}
For each channel realization $\bm{p}$ and ($\nframe'$, $\nslot'$) such that 
\begin{equation}
\begin{split}
\sum_{\islot=1}^\nslot {\rm I}^{\rm E}_{\iframe,\islot}(p_{\iframe,\islot}) < R_{{\rm d}, \iframe} \quad \forall \iframe =1, 2, \ldots, \nframe'-1 \\
\sum_{\islot=1}^{\nslot'} {\rm I}^{\rm E}_{\nframe,\islot}(p_{\nframe',\islot}) < R_{{\rm d}, \nframe'}
\end{split}
\label{conR2}
\end{equation}
and for each $n$, we have a $\delta(\epsilon)$ and a $\delta_\epsilon(n)$ such that $\delta_\epsilon(n) \xrightarrow[n \rightarrow \infty]{} 0$ and 
\begin{equation}
{\rm E}_{C_n} \left[\frac{1}{n} {\rm L}(C_n|\bm{p}, \nslot', \nframe') \right] \leq \nframe' \delta(\epsilon) + \nframe' \delta_\epsilon(n)\,.
\end{equation}
\end{lemma}
\begin{proof}
Follows from Lemmas \ref{l1} and \ref{l3}.
\end{proof}

\subsection{Characterization of the set $\mathcal P$}

Having derived sufficient conditions for decodability and secrecy for given channel realizations we are now ready to derive conditions for both decodability and secrecy with the same code. We now show that for the set 
\begin{equation}
\mathcal P =\{\bm{p}: \exists (\nframe', \nslot') \mbox{ for which both (\ref{conR1}) and (\ref{conR2}) hold}\}
\end{equation} 
there exists a single code (sequence) that provides both secrecy and decodability.

\begin{theorem}
\label{l6}
For all $n$ there exists a specific code $C^*_n$ with rates $R_s$ and $\{R_{{\rm d}, \iframe}\}$  such that, for all channels $\bm{p} \in \mathcal P$ there exists $\nframe'(\bm{p})$ and $\nslot'(\bm{p})$ such that 
\begin{equation}
\label{eq12}
{\rm P}_e(C_n^*|\bm{p}) \leq \delta_\epsilon(n)\,, \quad {\rm L}(C_n^*|\bm{p}, \nslot'(\bm{p}), \nframe'(\bm{p}))  \leq \delta(\epsilon) + \nframe \delta_\epsilon(n)
\end{equation}
and 
\begin{equation}
\lim_{n\rightarrow \infty} {\rm P}_e(C^*_n|\bm{p}) = 0\,,  \quad \lim_{n\rightarrow \infty} \frac{1}{n} {\rm L}(C^*_n|\bm{p}) \leq \delta(\epsilon)\,.
\label{sec_cond}
\end{equation}
\end{theorem}
\begin{proof}
From the definition of $\mathcal P$ and lemma \ref{l1} we immediately have
\begin{equation}
{\rm E}_{\bm{p} \in \mathcal P}[ {\rm E}_{C_n}[{\rm P}_e(C_n|\bm{p})]] \leq \delta_\epsilon(n)\,,
\end{equation}
while from Lemma \ref{l4} (and by the fact that $\nframe'(\bm{p}) < \nframe$) we also have
\begin{equation}
{\rm E}_{\bm{p} \in \mathcal P}\left[ {\rm E}_{C_n} \left[\frac{1}{n} {\rm L}(C_n|\bm{p}, \nslot'(\bm{p}), \nframe'(\bm{p})) \right]\right] \leq \nframe \delta(\epsilon) + \nframe \delta_\epsilon(n)\,.
\end{equation}
Then, similarly to the approach of \cite{Tang2009} we can swap the expectations over the channel set and the codes (since the integrands are non negative and finite), obtaining 
\begin{equation}
{\rm E}_{C_n}[{\rm E}_{\bm{p} \in \mathcal P}[ {\rm P}_e(C_n|\bm{p})]] \leq \delta_\epsilon(n)
\end{equation}
\begin{equation}
 {\rm E}_{C_n} \left[{\rm E}_{\bm{p} \in \mathcal P}\left[\frac{1}{n} {\rm L}(C_n|\bm{p}, \nslot'(\bm{p}), \nframe'(\bm{p})) \right]\right] \leq \nframe \delta(\epsilon) + \nframe \delta_\epsilon(n)\,.
\end{equation}
Now by applying the selection lemma \cite[pg. 14]{Bloch} to both functions ${\rm P}_e(\cdot)$ and ${\rm L}(\cdot)$, with reference to the random variable $C_n$, we obtain a sequence of codes with vanishing error probability and leakage. By observing that both ${\rm P}_e(\cdot)$ and ${\rm L}(\cdot)$ are non negative, we obtain (\ref{eq12}) and (\ref{sec_cond}).
\end{proof}

 We then have {\em a single code sequence} that provides both decodability and secrecy for {\em all channels in the set $\mathcal P$}.
 
\paragraph*{Remark 1} this result generalizes that of \cite{Tang2009}: for that code construction in fact, sufficient conditions for secrecy were ensured by a constraint on the sum of the mutual information between Alice and Eve across slots of a single frame. In our scenario instead we need bounds on each frame separately, as indicated by (\ref{conR2}).

\section{Secure channel sets}

Since the set $\mathcal P$ is defined in terms of the mutual information of the Alice-Bob and Alice-Eve channels, we can equivalently describe it by the set of mutual informations satisfying the constraints, i.e., by the set 
\begin{equation}
\begin{split}
&\mathcal Q = \{\{I_{\iframe,\islot}^{\rm B},I_{\iframe,\islot}^{\rm E}\}: \\
&\ \bm{p} \in \mathcal P, \mbox{ and } I_{\iframe,\islot}^{\rm B} = I_{\iframe,\islot}^{\rm B}(p_{\iframe,\islot})\,, I_{\iframe,\islot}^{\rm E} = I_{\iframe,\islot}^{\rm E}(p_{\iframe,\islot})\}\,.
\end{split}
\end{equation}
From the results of the previous Section we have 
\begin{equation}
{\mathcal Q}  = \bigcup_{\nframe'=1}^{\nframe}\bigcup_{\nslot'=1}^{\nslot} \left[ {\mathcal Q}^{\rm (E)}_{\rm S}(\nframe',\nslot') \cap {\mathcal Q}^{\rm (B)}_{\rm S}(\nframe',\nslot') \right],
\end{equation}
where ${\mathcal Q}^{\rm (E)}_{\rm S}(\nframe',\nslot')$ indicates the set of channels for which no information about the secret message has leaked to Eve up to the $\islot$-th slot of the $\iframe$-th frame, 
\begin{IEEEeqnarray*}{l}
{\mathcal Q}^{\rm (E)}_{\rm S}(\nframe',\nslot') = \left\{\{I_{\iframe,\islot}^{\rm B},I_{\iframe,\islot}^{\rm E}\} :\right. \\
\quad \sum_{\islot'=1}^{\nslot} I_{\iframe',\islot'}^{\rm E} \leq R_{{\rm d}, \iframe'}\,, \mbox{ for } \iframe' = 1, 2, \ldots, \nframe'-1 \,,\\
\left.\quad\sum_{\islot'=1}^{\nslot'} I_{\nframe',\islot'}^{\rm E} \leq R_{{\rm d}, \nframe'}\right\}\IEEEyesnumber
\label{sc}
\end{IEEEeqnarray*}
while ${\mathcal Q}^{\rm (B)}_{\rm S}(\nframe',\nslot')$ indicates the set of channels for which the secret message is decodable by Bob within the $\islot$-th slot of the $\iframe$-th frame, 
\begin{IEEEeqnarray}{l}
{\mathcal Q}^{\rm (B)}_{\rm S}(\nframe',\nslot') = \left\{\{I_{\iframe,\islot}^{\rm B},I_{\iframe,\islot}^{\rm E}\} : \right. \IEEEnonumber\\
\quad\sum_{\iframe'=1}^{\nframe'-1} \left[\sum_{\islot'=1}^{\nslot} I_{\iframe',\islot'}^{\rm B} - R_{{\rm d}, \iframe'}\right]^+ \IEEEnonumber\\
\quad {+}\: \left.\left[\sum_{\islot'=1}^{\nslot'} I_{\nframe',\islot'}^{\rm B} - R_{{\rm d}, \nframe'}\right]^+ \geq R_s \right\}
\label{condR5}
\end{IEEEeqnarray}
where $[x]^+ = 0$ if $x < 0$ and $[x]^+ = x$ if $x > 0$. Condition (\ref{condR5}) follows by applying Lemma \ref{l1} to the Alice-Bob channel, as we observe that if $\sum_{\islot=1}^{\nslot} I_{\iframe,\islot}^{\rm B} \leq R_{{\rm d}, \iframe} $, Bob will not make use of the signal received in frame $\iframe$ to decode the secret message.

\subsection{Outage Analysis}
\label{outage}

%In the previous section we have derived conditions on the channel that ensures secrecy and decodability of the message. Now we derive the probability that those conditions are not met over a fading channel, thus yielding a transmission outage.

A bound on the reliability outage probability for the whole transmission is then
\begin{equation}
P_{\rm  o} \leq {\rm P}\left[\{I_{\iframe,\islot}^{\rm B},I_{\iframe,\islot}^{\rm E}\}  \not\in\mathcal Q^{\rm (B)}(\nframe,\nslot)\right]\,,
\label{Po}
\end{equation}
while the probability that decoding happens exactly at the $\islot$-the slot of the $\iframe$-th frame is bounded as
\begin{equation}
\begin{split}
P_{\rm  D}(\iframe,\islot) \geq \\
\begin{cases}
{\rm P}\left[\{I_{\iframe,\islot}^{\rm B},I_{\iframe,\islot}^{\rm E}\}  \in\mathcal Q^{\rm (B)}(\iframe,\islot)\setminus\mathcal  Q^{\rm (B)}(\iframe,\islot-1)\right] & \islot > 1 \\
{\rm P}\left[\{I_{\iframe,\islot}^{\rm B},I_{\iframe,\islot}^{\rm E}\}  \in\mathcal Q^{\rm (B)}(\iframe,1)\setminus\mathcal  Q^{\rm (B)}(\iframe-1,\nslot)\right] & \islot = 1 
\end{cases}\,.
\end{split}
\label{PDkm}
\end{equation}
%Let $\nframe'$ be the effective number of retransmissions made by Alice. 
Assuming that the Alice-Bob and Alice-Eve channels are independent, the secrecy outage probability up to slot $\islot$ of frame $\iframe$ is bounded by
\begin{equation}
P_s(\iframe,\islot) \leq {\rm P}\left[\{I_{\iframe,\islot}^{\rm B},I_{\iframe,\islot}^{\rm E}\}  \not\in\mathcal Q^{\rm (E)}(\iframe,\islot)\right]\,.
\label{Ps}
\end{equation}

\section{Numerical Results}

We first provide some insight into the performance of the proposed solution by considering a transmission with only two frames $\nframe=2$ and one slot per frame $\nslot=1$. For given values of $R_{{\rm d}, 1}$, $R_{{\rm d}, 2}$ and $R_s$, Fig. \ref{fig2reg} shows as a dashed area the set $\mathcal Q_S^{\rm (B)}(2,1)$ with $\nframe=2$ and $\nslot=1$. We also show in gray the set $\mathcal Q_S^{\rm (B)}(1,2)$  with $\nframe=1$ frame and  $\nslot=2$ slot, that is when the \ac{IR-HARQ} scheme  is used with a random rate $R^{IR}_{\rm d}$ over the same channel. We observe that the shape of the two areas are different. Similarly, for given values of $R_{{\rm d}, 1}$, $R_{{\rm d}, 2}$ and $R_s$, Fig. \ref{fig3reg} shows as a dashed area the set $\mathcal Q_S^{\rm (E)}(2,1)$ with $\nframe=2$ and $\nslot=1$. We also show in gray the performance of \ac{IR-HARQ}. Also in this case we observe that the shape of the two areas are different. We conclude that \ac{S-HARQ} is a non-trivial extension of \ac{IR-HARQ}.

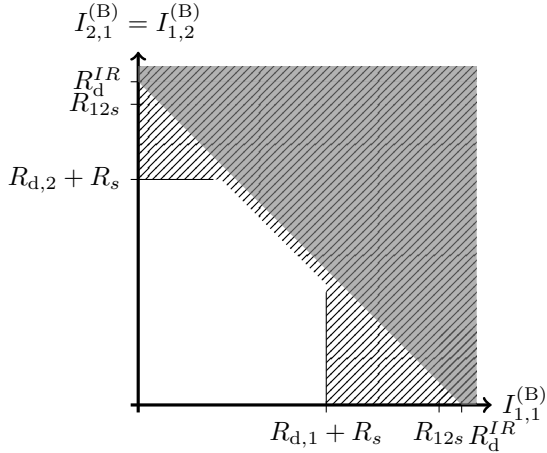
\begin{figure}
\centering
\begin{tikzpicture}

    %\draw[gray!50, thin, step=0.5] (-1,-3) grid (5,4);
    \draw[very thick,->] (-.1,0) -- (4.7,0) node[right] {$I_{1,1}^{\rm (B)}$};
    \draw[very thick,->] (0,-.1) -- (0,4.7) node[above] {$I_{2,1}^{\rm (B)} = I_{1,2}^{\rm (B)}$};

    %\foreach \x in {-1,...,5} \draw (\x,0.05) -- (\x,-0.05) node[below] {\tiny\x};
    %\foreach \y in {-3,...,4} \draw (-0.05,\y) -- (0.05,\y) node[right] {\tiny\y};
    \draw (0.05,4) -- (-0.1,4) node[left] {{{$R_{12s}$}}};
    \draw (4,0.05) -- (4,-0.1) node[below] { $R_{12s}$};
    \draw (0.05,4.3) -- (-0.1,4.3) node[left] {{{$R_{{\rm d}}^{IR}$}}};
    \draw (4.3,0.05) -- (4.3,-0.1) node[below] { \hspace{0.5cm} $\;\;R_{{\rm d}}^{IR}$};

    \draw (0.05,3) -- (-0.1,3) node[left] { $R_{{\rm d}, 2}+R_s$};
    \draw (2.5,0.05) -- (2.5,-0.1) node[below] { $R_{{\rm d}, 1}+R_s$};
    
    %\fill[blue,opacity=0.3] (0,4) -- (4,0) -- (4,4) -- cycle;
    \fill[pattern=north east lines] (0,4.1) -- (0,3) -- (1.1,3);
    \fill[pattern=north east lines] (2.5,1.6) -- (2.5,0) -- (4.1,0);
    \fill[pattern=north east lines] (0.1,4) -- (4.5,4) -- (4.5,0.1) -- (4,0.1)-- cycle;    
    \fill[pattern=north east lines] (4,0) -- (4,1) -- (4.5,1) -- (4.5,0);    
    \fill[pattern=north east lines] (0,4) -- (0,4.5) -- (4.5,4.5) -- (4.5,4);    

    \fill[gray,opacity=0.6] (0,4.3) -- (4.5,4.3) -- (4.5,0) -- (4.3,0); 
    \fill[gray,opacity=0.6] (0,4.3) -- (0,4.5) -- (4.5,4.5) -- (4.5,4.3);    
%    \fill[blue,opacity=0.1] (0,0) -- (0, 3) -- (1,3) -- (2.5,1.5) --  (2.5,0) -- cycle;
        
%    \draw[dashed] (0,4) --  (4,0);
    \draw (0,3) -- (1,3);    
    \draw (2.5,1.5) --  (2.5,0);    
    %\draw (1,-3) -- (3,1) -- node[below left,sloped] {\tiny$2x_1-x_2\leq5$} (4.5,4);
    %\draw (-1,1) -- node[above,sloped] {\tiny$-x_1+2x_2\leq3$} (5,4);
    
\end{tikzpicture}
\caption{$R_{12s} = R_{{\rm d}, 1}+R_{{\rm d}, 2}+R_s$ Bob's decoding region $\mathcal Q_S^{\rm (B)}(2,1)$ for \ac{S-HARQ} with $K=2$ and $M=1$ (dashed area), and $\mathcal Q_S^{\rm (B)}(1,2)$ of \ac{IR-HARQ} (gray area).}
\label{fig2reg}
\end{figure}

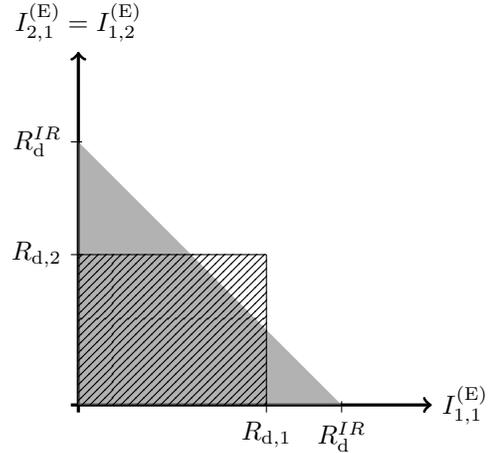
\begin{figure}
\centering
\begin{tikzpicture}

    %\draw[gray!50, thin, step=0.5] (-1,-3) grid (5,4);
    \draw[very thick,->] (-.1,0) -- (4.7,0) node[right] {{$I_{1,1}^{\rm (E)}$}};
    \draw[very thick,->] (0,-.1) -- (0,4.7) node[above] {$I_{2,1}^{\rm (E)} = I_{1,2}^{\rm (E)}$};

    %\fill[blue,opacity=0.9] (2.5,2) -- (2.5,1.5) -- (2,2) -- cycle;
    \fill[gray,opacity=0.6] (0,0) -- (0,3.5) -- (3.5,0) -- cycle;
    \fill[pattern=north east lines] (0,0) -- (0, 2) -- (2.5,2) --  (2.5,0) -- cycle;

    %\foreach \x in {-1,...,5} \draw (\x,0.05) -- (\x,-0.05) node[below] {\tiny\x};
    %\foreach \y in {-3,...,4} \draw (-0.05,\y) -- (0.05,\y) node[right] {\tiny\y};
    \draw (0.05,3.5) -- (-0.1,3.5) node[left] { $R_{{\rm d}}^{IR}$};
    \draw (3.5,0.05) -- (3.5,-0.1) node[below] { $R_{{\rm d}}^{IR}$};
    \draw (0.05,2) -- (-0.1,2) node[left] { $R_{{\rm d}, 2}$};
    \draw (2.5,0.05) -- (2.5,-0.1) node[below] { $R_{{\rm d}, 1}$};
    
    %\fill[blue!50!cyan,opacity=0.3] (8/3,1/3) -- (1,2) -- (13/3,11/3) -- cycle;

%    \draw[dashed] (0,4.5) --  (4.5,0);
    \draw (0,2) -- (2.5,2);    
    \draw (2.5,2) --  (2.5,0);    
    %\draw (1,-3) -- (3,1) -- node[below left,sloped] {\tiny$2x_1-x_2\leq5$} (4.5,4);
    %\draw (-1,1) -- node[above,sloped] {\tiny$-x_1+2x_2\leq3$} (5,4);

\end{tikzpicture}
\caption{ Eve's failure region $\mathcal Q_S^{\rm (E)}(2,1)$ for \ac{S-HARQ} with $K=2$ and $M=1$ (dashed area) and $\mathcal Q_S^{\rm (B)}(1,2)$ of \ac{IR-HARQ} (gray area).}
\label{fig3reg}
\end{figure}

\begin{figure}
\centering
\includegraphics[width=.8\hsize]{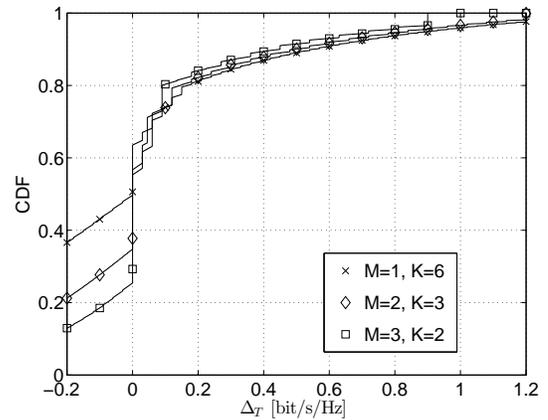}
\caption{\ac{CDF} of the difference between the throughputs of S-HARQ and IR-HARQ for various values of $\nslot$, over a Rice fading wiretap channel. }
\label{figres1}
\end{figure}

In order to further confirm this conclusion in a more general setting, we have considered \ac{S-HARQ} with a total of $6$ slots. Both Alice-Bob and Alice-Eve channels are block-fading, with independent Rice channels at each slot: the Rice factor is $0$\,dB for both and the average \ac{SNR} is $4$\,dB for the Alice-Bob channel and $5$\,dB, for the Alice-Eve channel. The performance of the proposed approach is assessed over block fading channels, by considering the achievable secret throughput, i.e.
\begin{equation}
\label{defT}
T = \max_{\{R_{\rm d,k}\}} \sum_{\iframe=1}^{\nframe} \sum_{\islot=1}^{\nslot} \frac{P_{\rm  D}(\iframe,\islot)[1-P_{\rm  s}(\iframe,\islot)] R_s}{\nslot(\iframe-1)+\islot}\,,
\end{equation}
which is the average (over the number of slots) of the secrecy rate divided by the number of slots needed for detection. As a reference value we consider the achievable secret throughput of the \ac{IR-HARQ} scheme $T^{\rm (IR)}$ and we focus on the additional secret throughput defined as $\Delta_T = T^{\rm (S)} - T^{\rm (IR)}$. Using (\ref{Po}), (\ref{PDkm}) and (\ref{Ps}) with equalities in (\ref{defT}) we obtain a lower bound on the achievable secret throughput. Note that we resort to the bounds since as exact performance is not known, although the difference of the lower bounds $\Delta_T$ may not in general be a bound of the actual difference. For various values of $\nslot$, Fig. \ref{figres1}  shows the \ac{CDF} of $\Delta_T$. Moreover, the total of 6 slots can be split into the following frame configurations: $a)$ $K=1$, $M=6$, $b)$ $K=2$, $M=3$, $c)$ $K=3$, $M=2$, $d)$ $K=6$, $M=1$.  Recall that $\nframe = 1, \nslot = 6$ corresponds to the IR-HARQ system. From the figure, we note that for all $\nslot < 6$, we have $\Delta_T > 0$ with non-zero probability, therefore there are cases when the diversity provided by different frames within the same total number of slots, yields a strictly higher throughput than choosing $\nframe=1$. Moreover, we observe that by varying $\nslot$ the distribution of $\Delta_T$ changes, thus leaving space for optimization of the system, to be considered in future studies.

\section{Conclusions}

We have proposed a secret message transmission scheme over block fading channels with a feedback from the legitimate receiver with no \ac{CSI}. Numerical results have highlighted the non-trivial relation with existing schemes and the fact that for some channel conditions the proposed solution provides a higher available secret throughput.

\section*{Acknowledgment}
This work was supported in part by the MIUR project ESCAPADE
(Grant RBFR105NLC) under the ``FIRB-Futuro in Ricerca 2010'' funding program.

\appendix
\label{proof_l5}

In this Appendix we provide the proof of Lemma \ref{l5}.

%\balance

For the sake of a simpler notation we provide the proof for $K' =K$ and $M'=M$, the generalization being straightforward. 

Without restriction we suppose $\mathcal M = i = 1$. Indicating with $T_{\nframe, \epsilon}^{nM}(X_\iframe Y_\iframe)$ the set of all  $\epsilon$-jointly weakly typical sequences $(\{X_1^{Mn}(i,j_1), \ldots, X_{\nframe}^{Mn}(i,j_\iframe)\}, \{Y_{1}^{Mn}, \ldots, Y_{\nframe}^{Mn}\})$, we can express ${\rm E}_{C_n}[{\rm P}_e(C_n|\bm{p})]$ in terms of the events
\begin{equation}
\mathcal E_{i,j_1, \ldots, j_\nframe} = \{\, (\{X_\iframe^{nM}(i,j_\iframe)\}, \{Y_\iframe^{nM}\}) \in \mathcal T_{\nframe, \epsilon}^{nM}(X_k Y_k)\,\} 
\label{defE}
\end{equation}
for $(i,j_1, \ldots, j_\nframe) \in [1, 2^{nR_s}] \times [1, 2^{nR_{{\rm d}, \iframe}}] \times  \cdots \times [1, 2^{nR_{{\rm d}, \iframe}}]$ as
\begin{equation}
\begin{split}
& {\rm E}_{C_n}[{\rm P}_e(C_n|\bm{p})] =  {\rm P}\left[ \mathcal E_{1,1, \ldots, 1}^c \cup  \right.
\\
& \left. \bigcup_{(1,j_1, \ldots, j_\nframe): \{\exists \iframe: j_\iframe \neq 1\}}  \mathcal E_{i,j_1, \ldots, j_\nframe} \cup \bigcup_{(i,j_1, \ldots, j_\nframe): i \neq 1}  \mathcal E_{i,j_1, \ldots, j_\nframe}\right] \\
\end{split}
% \cup  \mathcal F_{1, \ldots, 1}^c \cup \bigcup_{j_1, \ldots, j_\nframe \neq 1} \mathcal F_{j_1, \ldots, j_\nframe} \right]\,.
\label{Eunito}
\end{equation}
%\marker{Secondo me manca qualcosa, per questo ho messo minore uguale. Es: majority logic}

By the asymptotic equipartition property we have
\begin{equation}
{\rm P}\left[ \mathcal E_{1,1,\ldots, 1}^c \right] \leq \delta_\epsilon(n)\,.
\end{equation}

Indicating with with $T_{ \epsilon}^{nM}(X_\iframe Y_\iframe)$ the set of all  $\epsilon$-jointly weakly typical sequences $(X_\iframe^{Mn}(i,j_{\iframe}), Y_{\iframe}^{Mn})$, define the event
\begin{equation}
\mathcal F_{i,j,\iframe} = \{(X_\iframe^{nM}(i,j), Y_\iframe^{nM}) \in \mathcal T_{\epsilon}^{nM}(X_\iframe Y_\iframe)\} 
\end{equation}
for $(i,j,\iframe) \in [1, 2^{nR_s}] \times [1, 2^{nR_{{\rm d}, \iframe}}] \times [1, \nframe]$, then we have
\begin{equation}
{\rm P}\left[ \mathcal E_{i,j_1,\ldots, j_\nframe} \right] =\prod_\iframe {\rm P}\left[ \mathcal F_{i,j_\iframe,\iframe} \right]\,.
\label{pcap}
\end{equation}

We can split $\bigcup_{(1,j_1, \ldots, j_\nframe): \{\exists \iframe: j_\iframe \neq 1\}}  \mathcal E_{i,j_1, \ldots, j_\nframe}$ as the union of the events where an error occurs in at least one frame. Let $\mathcal S$ be the set of $\iframe$ for which error occurs, then we have
\begin{equation}
{\rm P}\left[\mathcal E_{i, j_\iframe\neq 1, k \in \mathcal S\,, j_u=1, u \in \mathcal K \setminus \mathcal S}\right] \leq 2^{-n(\sum_{k\in\mathcal S} {\rm I}(X_k^{Mn}; Y_k^{Mn}| p_{\iframe}) - \delta(\epsilon))}
\end{equation}
with $\mathcal K = \{1, 2, \ldots, \nframe\}$, and we have $2^{nR_{{\rm d}, \iframe}|\mathcal S|}-1$ of these events. Hence, the most restrictive condition on $R_{{\rm d}, \iframe}$ to have vanishing error probability is when only one of the $\nframe$ messages is different from 1 and in this case we have
\begin{equation}
{\rm P}\left[ \mathcal E_{1,j_1,\ldots, j_\nframe} \right] \leq 2^{-n({\rm I}(X_{\iframe}^{nM}; Y_{\iframe}^{nM}| p_{\iframe}) - \delta(\epsilon))}\,, \quad \exists ! \iframe: j_\iframe \neq 1
\end{equation}
and we have $2^{nR_{{\rm d}, \iframe}}-1$ of such events.

Event $ \mathcal E_{i,j_1, \ldots, j_\nframe}$, $i\neq 1$ occurs when $\mathcal F_{i,j,\iframe}$, $i \neq 1$ occurs for all $\iframe$, in which case from (\ref{pcap}) we have
\begin{equation}
{\rm P}\left[ \mathcal E_{1,j_1,\ldots, j_\nframe} \right] \leq 2^{-n(\sum_{\iframe=1}^\nframe {\rm I}(X_{\iframe}^{nM}; Y_{\iframe}^{nM}| p_{\iframe}) - \delta(\epsilon))}\,, \quad \exists ! \iframe: j_\iframe \neq 1
\end{equation}
and we have $(2^{nR_s}-1)2^{\sum_{\iframe=1}^\nframe R_{{\rm d}, \iframe}}$ of such events.

From (\ref{Eunito}) we have
\begin{equation*}
\begin{split}
& {\rm E}_{C_n}[{\rm P}_e(C_n|\bm{p})] \leq \\
\quad & \delta_\epsilon(n) + 
\sum_{\iframe=1}^\nframe (2^{nR_{{\rm d}, \iframe}}-1)  
2^{-n({\rm I}(X_{\iframe}^{nM}; Z_{\iframe}^{nM}| p_{\iframe})- \delta(\epsilon))} + \\
\quad & 2^{ \sum_{\iframe=1}^\nframe R_{{\rm d}, \iframe}} (2^{nR_s} -1) 2^{-n\sum_{\iframe=1}^\nframe ({\rm I}(X_{\iframe}^{nM}; Y_{\iframe}^{nM}| p_{\iframe}) - \delta(\epsilon))} \,.
\end{split}
\end{equation*}

Lastly, observing that 
\begin{equation}
{\rm I}(X_{\iframe}^{nM}; Y_{\iframe}^{nM}| p_{\iframe}) = \sum_{\islot=1}^{\nslot} {\rm I}(X_{\iframe,\islot}^n; Y_{\iframe,\islot}^n| p_{\iframe,\islot})\;,
\end{equation}
 we conclude the proof.

\bibliographystyle{IEEEtran}
\bibliography{biblio}

\end{document}